\documentclass[a4paper,11pt]{article}

\usepackage{amsmath}
\usepackage{amssymb}
\usepackage{amsthm}
\usepackage[OT4]{fontenc}
\usepackage[utf8]{inputenc}
\usepackage[margin=3.5cm]{geometry}
\usepackage[square,numbers]{natbib}
\usepackage{graphicx}
\usepackage{algorithm}
\usepackage[noend]{algpseudocode}
\usepackage{xurl}
\usepackage[unicode,debug,colorlinks=true,linkcolor=black,citecolor=black,urlcolor=black,pdfusetitle=true,breaklinks=true]{hyperref}
\allowdisplaybreaks

\newcommand{\abs}[1]{\left| #1 \right|}
\newcommand{\okra}[1]{\left( #1 \right)}

\newcommand{\klam}[1]{\left\{ #1 \right\}}
\newcommand{\boole}[1]{{\mathbf 1}{\klam{#1}}}
\newcommand{\ceil}[1]{\left\lceil #1 \right\rceil}
\newcommand{\floor}[1]{\left\lfloor #1 \right\rfloor}

\DeclareMathOperator{\card}{\#}

\DeclareMathOperator{\corr}{Corr}

\DeclareMathOperator{\mean}{\mathbf{E}}

\newtheorem{definition}{Definition}
\newtheorem{example}{Example}
\newtheorem{lemma}{Lemma}
\newtheorem{theorem}{Theorem}

\begin{document}

\title{Multiperiodic Processes: \\
  Ergodic Sources with a Sublinear Entropy}

\author{{\L}ukasz D\k{e}bowski%
  \thanks{{\L}ukasz D\k{e}bowski is with the Institute of Computer
    Science, Polish Academy of Sciences, ul. Jana Kazimierza 5, 01-248
    Warszawa, Poland, e-mail: ldebowsk@ipipan.waw.pl.}%
}

\date{}

\begin{titlepage}
\maketitle

\begin{abstract}
  Several explicit stochastic processes are known to satisfy Hilberg's
  law, a power-law growth of block entropy conjectured for natural
  language and recently connected to the neural scaling law. Existing
  examples either possess a positive Shannon entropy rate, are
  non-ergodic, or require comparatively involved constructions. We
  introduce multiperiodic processes, a new class of stationary ergodic
  processes over the natural numbers generated by random shifts of
  deterministic multiperiodic sequences. Under mild conditions,
  multiperiodic processes have vanishing Shannon entropy rate and,
  under a suitable parameterization, they satisfy both Zipf's law for
  symbol frequencies and Hilberg's law for block entropy. Since
  multiperiodic processes are not mixing, we identify the open problem
  of constructing an elementary strongly mixing source with vanishing
  entropy rate and Hilberg's law.
  \\[1em]
  \\[1em]
\end{abstract}


\end{titlepage}

\section{Introduction}
\label{secIntroduction}

Power-law growth of block entropy, commonly referred to as Hilberg's
law \cite{Hilberg90}, has long been conjectured to characterize
natural language and related symbolic data
\cite{EbelingNicolis91,CrutchfieldFeldman03,Debowski11b,TakahiraOther16}.
Recently, it was shown in \cite{Debowski25g,CagnettaOther26} that
Hilberg's law implies the neural scaling law, namely a power-law decay
of the cross entropy rate of a large language model as a function of
the amount of training resources. The neural scaling law has been
observed empirically for internet-scale text corpora
\cite{HestnessOther17, KaplanOther20, HenighanOther2020,
  HernandezOther21, HoffmannOther22, PearceSong24, LiOther25}.  These
developments motivate the search for explicit stochastic sources that
exhibit Hilberg-type entropy growth and can serve as analytically
tractable reference models in machine learning and quantitative
linguistics.

The notations to state our results are as follows. Blocks of random
variables are denoted by $X_j^k:=(X_j,X_{j+1},\ldots,X_k)$. Let
$H(X):=\mean\okra{-\log P(X)}$ be the Shannon entropy of a discrete
random variable $X$, where $\log$ denotes the base 2 logarithm. Let
$H(X|Y):=H(X,Y)-H(Y)$ be the conditional entropy and
$I(X;Y):=H(X)-H(X|Y)$ the mutual information. We write
\begin{align}
f(t)\sim g(t)
\quad\text{iff}\quad
\lim_{t\to\infty}\frac{\log f(t)}{\log g(t)}=1.
\end{align}
For a particular stationary process $(X_t)_{t\in\mathbb{Z}}$,
Hilberg's law is the relationship
\begin{align}
  \label{Hilberg}
  H(X_1^t)-ht\sim t^\beta,
\end{align}
where $h$ is the entropy rate \cite{Shannon48} and $\beta\in(0,1)$ is
the Hilberg exponent \citep{Hilberg90}.

Several explicit examples that obey law (\ref{Hilberg}) are already
known. In particular, the Santa Fe and Oracle processes
\cite{Debowski09,Debowski11b,Hutter21,Debowski21b}, to be recalled in
Section \ref{secHilberg}, provide simple examples of stationary
sources over a countable alphabet whose block entropy obeys a power
law with a positive entropy rate. However, as we will explain further
in Section \ref{secWhy}, it is not completely certain whether
condition $h>0$ holds for natural language.  To make an informed
opinion about this issue, we feel motivated to better understand the
mathematical construction of zero entropy rate processes with a
power-law block entropy. There are random hierarchical association
(RHA) processes \cite{Debowski17,Debowski21} which follow law
(\ref{Hilberg}) with $h=0$ but their construction is ridiculously
involved and their ergodic properties remain unclear. Definitely for
machine learning and quantitative linguistic experiments, we would
like to have a simpler kind of a baseline source with similar
properties.

The present paper is a step towards filling this gap. We introduce a
new elementary class of stationary ergodic sources, called
multiperiodic processes.  Multiperiodic processes
$(K_t)_{t\in\mathbb{Z}}$ are asymptotically deterministic sequences of
random natural numbers. They are ergodic but not mixing and, under an
appropriate parameterization, they satisfy Hilberg's law with
vanishing Shannon entropy rate,
\begin{align}
\label{HilbergZero}
H(K_1^t)\sim t^\beta.
\end{align}

This construction is considerably simpler the RHA processes
\cite{Debowski17,Debowski21}. The multiperiodic process
$(K_t)_{t\in\mathbb{Z}}$ is supported on a family of deterministic
multiperiodic sequences generated by a mechanism called the Infinite
Clock. Randomness in this process enters only through independent
random shifts. Concretely, the Infinite Clock algorithm, given certain
parameters $\pi_k\in\mathbb{N}$, called periods, and
$\sigma_k=1,2\ldots,\pi_k$, called seeds, returns a deterministic
infinite multiperiodic sequence $(k_t)_{t\in\mathbb{Z}}$. To obtain a
process $(K_t)_{t\in\mathbb{Z}}$, we introduce a sequence of
independent random seeds $(\Sigma_k)_{k\in\mathbb{N}}$ with the
uniform distributions $P(\Sigma_k=\sigma_k)=1/\pi_k$. Consequently, we
fix the random sequence
$(K_t)_{t\in\mathbb{Z}}=(k_t)_{t\in\mathbb{Z}}$ for event
$(\Sigma_k)_{k\in\mathbb{N}}=(\sigma_k)_{k\in\mathbb{N}}$.

Power laws in this asymptotically deterministic system arise under a
suitable choice of the period parameters. It suffices to take
$\pi_k=1+\ceil{ck}$ for a $c>0$, to obtain the relative frequencies
approaching Zipf's law
\begin{align}
\label{Zipf}
P(K_t=k)\sim k^{-\alpha},
\end{align}
where $\alpha=(c+1)/c>1$, cf.\ \cite{Zipf35,Mandelbrot54}. Although
the process is far from being a sequence of independent random
variables, it inherits a power-law growth of the vocabulary observed
for memoryless sources with the marginal law (\ref{Zipf}). Defining
the set of observed types as
\begin{align}
\label{TypesIntro}
  \mathcal{V}_t:=\klam{K_1,K_2,\ldots,K_t},
\end{align}
we obtain a power-law relationship
\begin{align}
  \label{Heaps}
  \mean \card\mathcal{V}_t\sim t^{1/\alpha},
\end{align}
known as Heaps' law in quantitative linguistics
\cite{Karlin67,Khmaladze88,Herdan64,Heaps78,Baayen01}.\footnote{We
  adopt the distinction of terms ``type'' and ``token'' borrowed from
  quantitative linguistics: A token is an element of a sequence with
  possible repetitions, whereas a type is an element of a set, which
  excludes repetitions.}

The key observations to derive entropic and ergodic properties of
multiperiodic processes are as follows. First, the block entropy
equals the entropy of the seeds associated with the types that appear
in a sample
\begin{align}
  \label{SeedsEntropy}
  H(K_1^t)=H(\klam{(k,\Sigma_k):k\in\mathcal{V}_t})
  \sim \mean \card\mathcal{V}_t \sim t^{1/\alpha},
\end{align}
because there is a remarkable one-to-one correspondence
$K_1^t\leftrightarrow \klam{(k,\Sigma_k):k\in\mathcal{V}_t}$. Hence
Hilberg's law holds with exponent $\beta=1/\alpha$ and $h=0$.

Second, unlike the latent variables in previously studied Santa Fe
processes \cite{Debowski09,Debowski11b,Hutter21,Debowski21b}, for
which we observe a formula similar to (\ref{SeedsEntropy}), the seeds
$(\Sigma_k)_{k\in\mathbb{N}}$ are not measurable with respect to the
shift-invariant algebra of $(K_t)_{t\in\mathbb{Z}}$. This distinction
allows multiperiodic processes to remain ergodic despite possessing an
infinite collection of independent seed variables.  These processes
are not mixing, however, since strictly periodic sequences with
randomized shifts are not mixing, either. --- This can be considered a
drawback of our construction as a model of real-life phenomena and
therefore examples relaxing this condition should be sought for.

Still, multiperiodic processes provide explicit examples of stationary
ergodic sources with vanishing entropy rate, power-law block entropy
growth, and long-range statistical structure. Together with Santa Fe
and related constructions, they furnish a family of analytically
tractable benchmark models whose entropy rate and Hilberg exponent can
be controlled \cite{Debowski15d}. Such elementary models may be useful
for studying the performance of entropy estimators, universal coding
schemes, and learning algorithms on data exhibiting long-range
dependencies. They also provide a mathematically simple setting for
investigating the relationship between entropy growth, predictability,
and power laws such as Zipf's law and the neural scaling law.

\paragraph{Organization of the article.}

In Section \ref{secContext}, we describe the research context. Section
\ref{secMultiperiodic} contains the results concerning multiperiodic
processes.  Section \ref{secConclusion} forms the conclusion.

\section{Research context}
\label{secContext}

The present paper belongs to a line of research that studies
stationary stochastic processes with long-range statistical
dependencies. A natural information-theoretic measure of such
dependencies is excess entropy, which equals the mutual information
between the past and future of a process.  Processes with infinite
excess entropy arise in several mathematical settings, ranging from
Gaussian models with slowly decaying correlations to symbolic
processes motivated by natural language and other complex data
sources. Among these examples, particular attention has been devoted
to processes satisfying Hilberg's law, a power-law growth of block
entropy that has been conjectured to characterize natural language and
was recently linked to the neural scaling law.  The purpose of this
section is to place multiperiodic processes within this broader
landscape.

\subsection{Excess entropy}
\label{secExcess}

For a stationary process $(X_t)_{t\in\mathbb{Z}}$, let us consider two
limits
\begin{align}
  h&:=\lim_{t\to\infty} \frac{H(X_1^t)}{t}
  =\lim_{t\to\infty} H(X_1|X_{-t}^0),
  \\
  E&:=\lim_{t\to\infty} \okra{H(X_1^t)-ht}
  =\lim_{t\to\infty} I(X_1^t;X_{-t}^0).
\end{align}
Limit $h$ is called the entropy rate \cite{Shannon48, CoverThomas06},
whereas $E$ is called the excess entropy \cite{CrutchfieldFeldman03}
or the predictive information \cite{BialekNemenmanTishby01,
  BialekNemenmanTishby01b, FutrellHahn25}.  While entropy rate is a
measure of process unpredictability, excess entropy is an intuitive
measure of process memory. Although excess entropy only summarizes the
memory capacity, without specifying how exactly the process future
depends on the past, it can be given interesting interpretations. Some
attempts were also made to generalize excess entropy to
two-dimensional random fields
\cite{FeldmanCrutchfield03,BulatekKaminski09}.

\paragraph{Ergodic decomposition.}

A central question is whether the excess entropy is finite or
infinite. The distinction has appeared in several areas of probability
theory, information theory, and statistical physics.  The basic
observation is that condition $E=\infty$ holds if the process is
strongly non-ergodic in the sense that its shift-invariant algebra
$\mathcal{I}$ is non-atomic. It is so because
\begin{align}
  \label{ExcessDecomposition}
  E=H(\mathcal{I})+I(X_{-\infty}^0;X_1^\infty|\mathcal{I})
\end{align}
and $H(\mathcal{A})=\infty$ for a non-atomic algebra $\mathcal{A}$
under an appropriate generalization of Shannon information measures to
arbitrary algebras of events
\cite{GelfandKolmogorovYaglom56en,Dobrushin59en,Pinsker60en,
  Wyner78,Debowski09,Debowski20}. However, infinite excess entropy can
arise also for ergodic and mixing sources.

\paragraph{Gaussian processes.}

The earliest systematic studies of excess entropy appear to concern
Gaussian processes.  Grenander and Szeg\H{o} \cite[Section
5.5]{GrenanderSzego84} presented an integral formula for excess
entropy (in disguise). Finch \cite{Finch60} evaluated this formula for
autoregressive moving average (ARMA) processes, which turned out to
yield $E<\infty$. More generally, we have $E<\infty$ if and only if
$\sum_{k=1}^\infty k\alpha_k^2<\infty$ and $\abs{\alpha_k}<1$ for the
partial autocorrelation $\alpha_k$ \cite{Debowski05en}. In particular,
under a positive and continuous spectral density, Li \cite{Li06}
showed that $E<\infty$ if and only if the autocorrelation
$\rho_k:=\corr(X_0;X_k)$ satisfies
$\sum_{k=1}^\infty k\rho_k^2<\infty$. Thus, for Gaussian processes,
infinite excess entropy is closely related to slowly decaying
correlations and more standard notions of long-range dependence
\cite{Beran94}.

\paragraph{Discrete examples.}

For discrete-valued processes, the situation is richer. By the
data-processing inequality, hidden Markov models (HMM) with finitely
many latent states necessarily satisfy $E<\infty$
\cite{CrutchfieldFeldman03}. Excess entropy can also be expressed in
terms of predictive and retrodictive $\epsilon$-machines, which are
minimal unifilar hidden Markov representations of the process
\cite{ShaliziCrutchfield01,Lohr09,EllisonMahoneyCrutchfield09,
  MahoneyEllisonCrutchfield09}.  Infinite excess entropy first
appeared in more exotic examples.  Bradley \cite{Bradley80}
constructed a mixing process with $E=\infty$, while Gramss
\cite{Gramss94} studied a process generated from frequencies in the
binary rabbit sequence. Later, Travers and Crutchfield
\cite{TraversCrutchfield11} and D\k{e}bowski \cite{Debowski14}
constructed hidden Markov models with countably many latent states
that satisfy $E=\infty$. These examples also obey the power-law growth
of block entropy (\ref{Hilberg}) to be discussed in the next
subsection.

\subsection{Hilberg's law}
\label{secHilberg}

A more systematic study of infinite excess entropy for the discrete
alphabet case was inspired by the publication by Hilberg
\cite{Hilberg90} who redrew the plot by Shannon \cite{Shannon51} in
the log-log scale and proposed that relationship (\ref{Hilberg}) with
$\beta\approx 1/2$ and $h\approx 0$ holds for natural language.  The
first reception of this power law came from physicists
\cite{EbelingNicolis91,EbelingPoschel94,
  BialekNemenmanTishby01b,CrutchfieldFeldman03,Lohr09}.  By contrast,
we will show that Hilberg's law (\ref{Hilberg}) can be easily related
to general quantitative linguistic observations, such as Zipf's law
(\ref{Zipf}).

\paragraph{Santa Fe processes.}

The ergodic decomposition of excess entropy
(\ref{ExcessDecomposition}) inspires a family of linguistically
motivated examples of processes that obey law (\ref{Hilberg}). These
examples are called Santa Fe processes
\cite{Debowski05en,Debowski09,Debowski11b} and were independently
studied by Hutter \cite{Hutter21} as a model that yields the neural
scaling law. The construction of the Santa Fe process
$(X_t)_{t\in\mathbb{Z}}$ is to decompose each token $X_t$ as a pair of
a natural number $K_t$ and an additional bit --- which is copied from a
certain binary sequence $(Z_k)_{k\in\mathbb{N}}$ by taking the item at
position $K_t$. In other words, each text token $X_t$ may be written
as a pair
\begin{align}
  \label{SantaFe}
  X_t=(K_t,Z_{K_t}),
\end{align}
where $(Z_k)_{k\in\mathbb{N}}$, called knowledge, is the sequence of
random bits and $(K_t)_{t\in\mathbb{Z}}$, called narration, is the
sequence of random natural numbers.  In the default example, narration
$(K_t)_{t\in\mathbb{Z}}$ is a memoryless source with the marginal Zipf
distribution (\ref{Zipf}), whereas knowledge $(Z_k)_{k\in\mathbb{N}}$
is a sequence of independent fair coin flips, independent of process
$(K_t)_{t\in\mathbb{Z}}$. Under these conditions, we obtain Hilberg's
law (\ref{Hilberg}) with $\beta=1/\alpha$ because
\begin{align}
  H(X_1^t)-ht
  =H(\klam{(k,Z_k):k\in\mathcal{V}_t}|K_1^t)
  =\mean \card\mathcal{V}_t
  \sim t^{1/\alpha}
  \label{FactsEntropy}
\end{align}
for the set of types (\ref{TypesIntro}).  This process
$(X_t)_{t\in\mathbb{Z}}$ is non-ergodic and takes values in a
countably infinite alphabet. The reason for its non-ergodicity is that
coin flips $(Z_k)_{k\in\mathbb{N}}$ are measurable with respect to the
shift-invariant algebra of process $(X_t)_{t\in\mathbb{Z}}$. Thus each
realization of sequence $(Z_k)_{k\in\mathbb{N}}$ defines a distinct
ergodic component \cite{Debowski09,Debowski21}.

\paragraph{Oracle processes.}

The condition of infinite alphabet is not necessary to obtain
Hilberg's law.  Coding natural numbers as binary strings, we can
obtain Santa Fe-like processes over a finite alphabet that obey law
(\ref{Hilberg}) \cite{Debowski11b,Debowski10}. A particularly simple
instance are the Oracle processes described in \cite{Debowski21b},
which apply the monkey-typing explanation of Zipf's law
\cite{Mandelbrot54, Miller57}. Additionally, the ergodic components of
Oracle processes are unifilar hidden Markov models with a countably
infinite number of hidden states.

\paragraph{Mixing processes.}

Replacing individual coin flips $Z_k$ with slowly evolving
time-\-homo\-geneous binary Markov chains $(Z_{kt})_{t\in\mathbb{Z}}$ and
putting
\begin{align}
  \label{MixingSantaFe}
  X_t=(K_t,Z_{K_t,t}),
\end{align}
we obtain a process which is strongly mixing, and hence ergodic
\cite{Debowski12, Bradley05}. If the transition probabilities in
Markov chains $(Z_{kt})_{t\in\mathbb{Z}}$ are sufficiently small then,
after a lengthy calculation, we still recover Hilberg's law
(\ref{Hilberg}) for Zipf's law (\ref{Zipf}) \cite{Debowski12}. Here we
notice a time-\-inhomo\-geneous modification of this process which is
considerably simpler to analyze. Namely, it suffices to assume that
bits $Z_{k,t}$ flip randomly only for $t$ such that $K_t=k$ and do it
with a fixed probability $p\in(0,1)$. Such a process
$(X_t)_{t\in\mathbb{Z}}$ is still mixing but we obtain
\begin{align}
  H(X_1^t)-ht
  =I(p)\mean \card\mathcal{V}_t\sim t^{1/\alpha},
\end{align}
where $I(p):=1+p\log p+(1-p)\log(1-p)$ is the mutual information
between bits $Z_{k,t}$ observed at two consecutive times $t$ such that
$K_t=k$. This process exhibits an infinite but lossy memory and is
also mixing.

\paragraph{RHA processes.}

All the above examples enjoy Hilberg's law (\ref{Hilberg}) with a
positive entropy rate $h>0$, which is a function of parameter
$\alpha>1$. To complete the collection of invented sources, we have
been interested in finding an example of a process with Hilberg's law
(\ref{Hilberg}) and the vanishing Shannon entropy rate, $h=0$. In
\cite{Debowski17,Debowski21}, we have constructed the random
hierarchical association processes (RHA), which satisfy
(\ref{Hilberg}) with $h=0$ but are probably non-ergodic. The
construction of RHA processes is so complicated that we do not
reproduce it here. In the end, we desired a simpler construction.

\subsection{Why do such examples matter?}
\label{secWhy}


The construction of explicit stochastic processes exhibiting Hilberg's
law is motivated by several questions arising in information theory,
linguistics, and machine learning. Such processes provide analytically
tractable reference models for studying the interplay between entropy
rate, long-range statistical dependencies, and vocabulary growth.
Although highly idealized, they can serve both as thought experiments
and as benchmark data sources for numerical investigations. Their
potential applications range from testing entropy estimators,
universal compression algorithms, and language models to exploring
connections among Hilberg's law, Zipf's law, Heaps' law, and the neural
scaling law.

The motivations discussed below are not equally established. Some are
supported by empirical evidence, whereas others remain speculative.
Nevertheless, together they illustrate why mathematically simple
examples of processes with power-law entropy growth may be worth
studying.

\paragraph{Uncertainty about the entropy rate.}

The hypothesis that the entropy rate of natural language may vanish
does not lie in the mainstream of language studies but it does not
constitute a new concern.  Although researchers generally suppose that
the entropy rate of natural language is positive
\cite{TakahiraOther16, HoffmannOther22, PearceSong24, LiOther25},
Hilberg \cite{Hilberg90} put forward that law (\ref{Hilberg}) holds
for natural language with $h=0$.  There are also more recent papers
suggesting $h=0$.  These involve certain large-scale estimates of the
cross entropy of language models \cite{KaplanOther20} and the
empirical cube-logarithmic growth of the maximal repetition length
\cite{Debowski15f}.

Claims of $h=0$ may be spurious for the following reasons.  The
original hypothesis by Hilberg was based on a meager evidence, namely,
Shannon's guessing estimates of conditional entropy for $t\le 100$
letters redrawn in a log-log plot \cite{Shannon51}.  As for the
entropy rate estimate by Kaplan et al.\ \cite{KaplanOther20}, an
important issue is that fitting function $f(x)=ax^b+c$ to empirical
data can be numerically unstable.  Moreover, it has been proved
formally that the cube-logarithmic growth of the maximal repetition
length claimed in \cite{Debowski15f} implies the vanishing conditional
R\'enyi entropy rate rather than the vanishing Shannon entropy rate
\cite{Debowski18b, Debowski25e}.

To reduce potential doubt and to make an informed
opinion about the value of the Shannon entropy rate for natural
language, we wish to better understand the mathematical construction
of zero entropy rate processes with Hilberg's law.

\paragraph{Neural scaling law.}

The empirical neural scaling law describes a power-law decay of the
cross entropy rate of a large language model as a function of the
amount of training resources. It is observed for trillion-token
corpora of texts \cite{HestnessOther17, KaplanOther20,
  HenighanOther2020, HernandezOther21, HoffmannOther22, PearceSong24,
  LiOther25}. There is a growing body of literature in machine
learning that seeks to explain the neural scaling law. Some of these
works involve sophisticated mathematical frameworks, including
applications of random matrix theory and techniques from theoretical
physics such as Feynman diagrams \cite{MaloneyRobertsSully22,
  RobertsOther22}.

However, it has been also recognized \cite{Hutter21,
  MaloneyRobertsSully22, MichaudOther23, CabannesDohmatobBietti24,
  NeumannGros25, PanOther25} that the neural scaling law may arise as
a consequence of Zipf's law or similar distributional properties of
natural language. In particular, it was shown in
\cite{Debowski25g,CagnettaOther26} that Hilberg's law implies the
neural scaling law. These developments motivate the search for simple
stochastic sources that exhibit Hilberg's law and can serve as
reference models.

\paragraph{Semantic interpretation.}

The Santa Fe decomposition (\ref{SantaFe}) admits an appealing
semantic interpretation, discussed in
\cite{Debowski11b,Debowski21}. The pairs $(K_t,Z_{K_t})$ may be viewed
as elementary statements describing an underlying binary sequence
$(Z_k)_{k\in\mathbb{N}}$. These statements are mutually consistent:
whenever two tokens refer to the same index $k$, they assert the same
value $Z_k$. In this sense, the sequence $(Z_k)_{k\in\mathbb{N}}$ may
be interpreted as a repository of knowledge, whereas narration
$(K_t)_{t\in\mathbb{Z}}$ describes the process of selecting and
communicating individual pieces of that knowledge.  Finite texts then
reveal only a small fraction of an unbounded body of persistent
information.

This interpretation suggests connections with algorithmic information
theory \cite{LiVitanyi08}, see also \cite[Chapter 1]{Debowski21}. The
random quantity
\begin{align}
  Y:=\sum_{k=1}^{\infty}2^{-k}Z_k
\end{align}
encodes the entire sequence $(Z_k)_{k\in\mathbb{N}}$ into a single
real number in range $(0,1)$. In this respect, it resembles Chaitin's
halting probability $\Omega$ \cite{Chaitin05,Chaitin13}, whose binary
expansion compactly encodes a large amount of mathematical
information.  The analogy should not be taken too literally.  Whereas
$Y$ is a random variable, $\Omega$ is a fixed algorithmically random
number \cite{MartinLof66}. Moreover, the information encoded by
$\Omega$ appears fundamentally inaccessible in a way that differs from
the probabilistic knowledge represented by $(Z_k)_{k\in\mathbb{N}}$
\cite{Barmpalias18}. Nevertheless, both examples illustrate how a
large body of information may be represented through an apparently
simple underlying object.

\paragraph{Question of mixing.}

The semantic interpretation becomes more nuanced for the mixing Santa
Fe decomposition (\ref{MixingSantaFe}). In that model, the bits
$(Z_{kt})_{t\in\mathbb{Z}}$ evolve over time and therefore represent
changing rather than immutable knowledge. One may think of them as
describing conventions, beliefs, or facts about a changing
environment. Such evolving knowledge may be a more realistic model of
human communication than the immutable knowledge represented by
(\ref{SantaFe}).  From a mathematical perspective, the distinction
between (\ref{SantaFe}) and (\ref{MixingSantaFe}) is significant.
Allowing variables $(Z_{kt})_{t\in\mathbb{Z}}$ to change in a specific
way produces strongly mixing processes while preserving Hilberg's law
for suitable choices of the transition probabilities. If the rates of
change are chosen differently, however, the excess entropy may become
finite.  Thus the persistence of information and the mixing properties
of the process are closely related.

A similar interpretation can be formulated for the multiperiodic
processes introduced in this paper. In view of formula
(\ref{SeedsEntropy}), which parallels (\ref{FactsEntropy}), the seed
variables $(\Sigma_k)_{k\in\mathbb{N}}$ may again be viewed as a form
of latent immutable knowledge. Unlike the bits
$(Z_k)_{k\in\mathbb{N}}$ in the original Santa Fe construction,
however, the seeds are not measurable with respect to the
shift-invariant algebra of the process. Consequently, the distinction
between immutable knowledge and shift-invariant information becomes
unexpectedly visible.  However, multiperiodic processes remain
non-mixing because of their rigid periodic structure. For this reason
they are unlikely to serve as useful models of natural language. Yet,
they raise an interesting mathematical question: does there exist a
simple class of strongly mixing processes that simultaneously exhibits
Hilberg's law and vanishing entropy rate?  The present construction
does not answer this question, but it helps clarify the structural
obstacles that such an example would need to overcome.

\subsection{Road map of examples}
\label{secMap}

After these interdisciplinary divagations, which may be inspiring also
for strictly mathematical research, we would like to present a graphic
summary.  The summary takes form of a table describing discussed
stochastic processes and their properties. The checkmark (\checkmark)
indicates that a given property can be achieved for certain
parameters.
\begin{center}
  \bigskip
  \begin{tabular}{lccccc}
    class  & alphabet
    & $h=0$ & Hilberg's law & mixing
    & reference 
    \\
    \hline
    Gaussian & continuum
    & \checkmark & \checkmark & \checkmark
    & \cite{Debowski05en, Li06}
    \\
    finite-state HMM & finite
    & no & no & \checkmark
    & \cite{CrutchfieldFeldman03}
    \\
    countable-state HMM & finite
    & no & \checkmark & \checkmark
    & \cite{TraversCrutchfield11, Debowski14}
    \\
    default Santa Fe & countable
    & no & \checkmark & no
    & \cite{Debowski09,Debowski12,Hutter21}
    \\
    mixing Santa Fe & countable
    & no & \checkmark & \checkmark
    & \cite{Debowski12}
    \\
    default Oracle & finite
    & no & \checkmark & no
    & \cite{Debowski21b}
    \\
    default RHA & finite
    & \checkmark & \checkmark & ?
    & \cite{Debowski17,Debowski21}
    \\
    multiperiodic & countable
    & \checkmark & \checkmark & no
    & this work
    \end{tabular}
    \bigskip
\end{center}
As we can see, there is still no known class of simple discrete
examples that enjoys three checkmarks in the table. This open problem
is left for future investigation. Now let us proceed to the proper
results of this paper.



\section{Multiperiodic sequences and processes}
\label{secMultiperiodic}


This section develops the theory of multiperiodic processes in four
steps. First, we introduce multiperiodic sequences and the associated
stationary processes, establishing also their basic structural
properties. Second, we investigate their statistical behavior through
the distribution of symbols, waiting times, and vocabulary
growth. Third, we analyze the flow of information by studying the
estimation of hidden seeds from finite samples and the resulting block
entropy.  Finally, we illustrate the general theory using two concrete
families of examples based on constant and linear period sequences.

\subsection{Definition}
\label{secDefinitions}

We begin by defining multiperiodic sequences and the associated
multiperiodic processes. Along the way, we demonstrate several
fundamental properties that will be used throughout the paper,
including time-reversal symmetry, periodicity of decimation, existence
conditions, stationarity, and ergodicity.

\subsubsection{Multiperiodic sequences}
\label{secSequences}

Multiperiodic sequences are an idea which we have originally
generalized from work \cite{KalocinskiSteifer19}. These authors only
considered a fixed multiperiodic sequence with particular periods and
seeds --- to state a certain negative result in universal
prediction. Isolated instances of multiperiodic sequences are noted in
the on-line encyclopedia of integer sequences \cite{OEIS23} and can be
connected to other problems, such as Mancala-type games
\cite{BrolineLoeb95}. Distilling the general construction, especially
considering the idea of random seeds that matters in the probabilistic
context, seems our contribution.
\begin{definition}[multiperiodic sequence]
  \label{defiSequence}
  Consider a sequence $(\pi_k)_{k\in\mathbb{N}}$ of natural numbers
  $\pi_k\ge 2$. Let $(\sigma_k)_{k\in\mathbb{N}}$ be another sequence
  of natural numbers $\sigma_k\in\klam{1,2,\ldots,\pi_k}$ such that
  set $\klam{k\in\mathbb{N}:\sigma_k=1}$ is infinite. Parameters
  $\pi_k$ are called periods, whereas parameters $\sigma_k$ are called
  seeds.  Given $(\pi_k)_{k\in\mathbb{N}}$ and
  $(\sigma_k)_{k\in\mathbb{N}}$, the multiperiodic sequence
  $(k_t)_{t\in\mathbb{Z}}$ is the result of the initial condition
  \begin{align}
    \label{Initial}
    \phi_{k,1}=\sigma_k
  \end{align}
  and recursion
  \begin{align}
    \label{RecursionFirst}
    k_t
    &=\min\klam{k\in\mathbb{N}:\phi_{k,t}=1},
    \\
    \phi_{k,t+1}
    &=
      \begin{cases}
        \phi_{k,t}-1, & k<k_t,
        \\
        \pi_k, & k=k_t,
        \\
        \phi_{k,t}, & k>k_t,
      \end{cases}
  \end{align}
  which can be inverted as
  \begin{align}
    k_t
    &=\min\klam{k\in\mathbb{N}:\phi_{k,t+1}=\pi_k},
    \\
    \phi_{k,t}
    &=
      \begin{cases}
        \phi_{k,t+1}+1, & k<k_t,
        \\
        1, & k=k_t,
        \\
        \phi_{k,t+1}, & k>k_t.
      \end{cases}
      \label{RecursionLast}
  \end{align}
  The auxiliary matrix $(\phi_{k,t})_{k\in\mathbb{N},t\in\mathbb{Z}}$
  is called the hand matrix.
\end{definition}

Here we present two examples of multiperiodic sequences.
\begin{example}[A001511 in \cite{OEIS23}]
  Consider the multiperiodic sequence $(k_t)_{t\in\mathbb{N}}$ with
  constant periods $\pi_k=2$ and seeds $\sigma_k=1$, namely,
  \begin{align}
    1,2,1,3,1,2,1,4,1,2,1,3,1,2,1,5,1,2,1,3,1,2,1,4,1,2,1,\ldots  
  \end{align}
  Equivalently, $k_t$ is the number of $2$'s dividing $2t$.  The
  equivalence stems from the observation that $k_t$ is the position in
  the binary expansion of $t$ that is incremented when increasing from
  $t-1$ to $t$.  This sequence was applied by Kaloci\'nski and Steifer
  \cite{KalocinskiSteifer19} for a negative result in universal
  prediction.
\end{example}
Our applications are closer to this example:
\begin{example}[A028920 in \cite{OEIS23}]
  Let $\pi_k=1+k$ and $\sigma_k=1$ for all $k\in\mathbb{N}$. Then the
  multiperiodic sequence $(k_t)_{t\in\mathbb{N}}$ is
  \begin{align}
    1,2,1,3,1,4,1,2,1,5,1,6,1,2,1,3,1,7,1,2,1,8,1,4,1,2,1,\ldots
  \end{align}
  This sequence was discussed by Broline and Loeb \cite{BrolineLoeb95}
  in the context of Mancala-type games. The algorithm for generating
  this sequence is isomorphic with recursion
  (\ref{Initial})--(\ref{RecursionLast}).
\end{example}
\noindent
As we show in Section \ref{secLinear}, for periods $\pi_k\approx ck$,
the relative frequency of number $k$ approaches Zipf's law
(\ref{Zipf}) with $\alpha=(c+1)/c$.

Recursion (\ref{Initial})--(\ref{RecursionLast}) for $k\in\mathbb{N}$
can be rewritten as Algorithm \ref{algMultiperiodic}, called the
Infinite Clock.  The Infinite Clock operates with infinite data
structures for simplicity of presentation. It works like a clock with
infinitely many hands --- indexed by natural numbers --- that move
counterclockwise, hence the name. As in the usual clock, each hand
moves with a different speed, dictated by its specific period. The
initial positions of hands are given by seeds. The algorithm outputs
the value of the hand index whenever a hand passes noon.

\begin{algorithm}[h]
  \caption{\label{algMultiperiodic} The Infinite Clock}
\begin{algorithmic}[1]
  \Require{Sequence of $\pi_k\in\mathbb{N}$ for $k\in\mathbb{N}$.}
  \Comment{The periods.}
  \Require{Sequence of $\phi_k:=\sigma_k\in\klam{1,2,\ldots,\pi_k}$
    for $k\in\mathbb{N}$.}
  \Comment{The hands.}
  \Ensure{Sequence of $k_t\in\mathbb{N}$ for $t\in\mathbb{N}$.}
  \Comment{The multiperiodic sequence.}
  \For{$t\in\mathbb{N}$}
  \State{$k^*:= 0$}
  \State{$k:= 1$}
  \While{$k^*=0$}
  \If{$\phi_k>1$}
  \State{$\phi_k:=\phi_k-1$}
  \Else
  \State{$k^*:= k$}
  \EndIf
  \State{$k:= k+1$}
  \EndWhile
  \State{$\phi_{k^*}:= \pi_{k^*}$}
  \State{$k_t:= k^*$}
  \EndFor
\end{algorithmic}
\end{algorithm}

To be precise, Algorithm \ref{algMultiperiodic} operates with hand
variables $(\phi_k)_{k\in\mathbb{N}}$, where each hand $\phi_k$ takes
values in set $\klam{1,2,\ldots,\pi_k}$ and is initialized with the
value of the corresponding seed, $\phi_k:=\sigma_k$. To determine the
value of the next element $k_t$ of the sequence, the algorithm
decrements consecutive hands $(\phi_k)_{k\in\mathbb{N}}$ starting from
index $k=1$ until it finds an index $k^*$ such that
$\phi_{k^*}=1$. Then the loop is broken, this particular hand is reset
to its own maximal value $\phi_{k^*}:=\pi_{k^*}$, and the
corresponding index $k^*$ of the hand is output as the next element of
the sequence $k_t:=k^*$.

Let us observe that the computation falls into an infinite loop in
lines 4--9 of Algorithm \ref{algMultiperiodic} if and only if there is
no $k\in\mathbb{N}$ such that $\phi_k=1$.  This happens precisely if
and only if set $\klam{k\in\mathbb{N}:\sigma_k=1}$ is finite, which is
excluded by Definition \ref{defiSequence}.  The easiest way to prevent
this situation is to set the minimal seeds $\sigma_k=1$ for all
$k\in\mathbb{N}$.

As a general fact, we notice a time-reversal symmetry.
\begin{theorem}
  For the multiperiodic sequence $(k_t)_{t\in\mathbb{Z}}$ with periods
  $\pi_k$ and seeds $\sigma_k$, sequence $(k_{1-t})_{t\in\mathbb{Z}}$
  is the multiperiodic sequence with periods
  $(\pi_k)_{k\in\mathbb{N}}$ and seeds $(\sigma^R_k)_{k\in\mathbb{N}}$
  defined as $\sigma^R_k:=\pi_k-\sigma_k+1$.
\end{theorem}
\begin{proof}
  Let $(\phi^R_{k,t})_{k\in\mathbb{N},t\in\mathbb{Z}}$ and
  $(k^R_t)_{t\in\mathbb{Z}}$ be defined analogously as
  (\ref{RecursionFirst})--(\ref{RecursionLast}) but with the initial
  condition $\phi^R_{k,1}=\sigma^R_k$. Then by induction we obtain
  $\phi^R_{k,t}=\pi_k-\phi_{k,1-t}+1$ and $k^R_t=k_{1-t}$.
\end{proof}

Moreover, we observe a partial periodicity of decimated multiperiodic
sequences, which are also multiperiodic.
\begin{theorem}
  \label{theoDecimated}
  Let $(k^{\ge r}_t)_{t\in\mathbb{N}}$ be the subsequence of the
  multiperiodic sequence $(k_t)_{t\in\mathbb{N}}$ with periods $\pi_k$
  and seeds $\sigma_k$ from which we have removed all tokens
  $k_t<r$. We claim that:
  \begin{enumerate}
  \item Sequence $(k^{\ge r}_t-r+1)_{t\in\mathbb{N}}$ is the
    multiperiodic sequence with periods $\pi^{\ge r}_k=\pi_{k-r+1}$ and
    seeds $\sigma^{\ge r}_k=\sigma_{k-r+1}$.
  \item We have $k^{\ge r}_t=r \iff t\equiv \sigma_r\mod \pi_r$.
  \end{enumerate}
\end{theorem}
\begin{proof}
  We proceed by induction on $r$. As for the initial step, we
  have $(k^{\ge 1}_t)_{t\in\mathbb{N}}=(k_t)_{t\in\mathbb{N}}$. Thus
  claim 1.\ holds by definition, whereas claim 2.\ follows since
  $k_t=1$ if and only if $t\equiv \sigma_1\mod \pi_1$. As for the
  inductive step, it is sufficient to show that claim 1.\ holds,
  whereas claim 2.\ follows by claim 1.\ for the inductive step
  combined with claim 2.\ for the initial step. The proof of claim 1.\
  for the inductive step rests on the fact that to obtain
  $(k^{\ge r}_t)_{t\in\mathbb{N}}$ instead of
  $(k_t)_{t\in\mathbb{N}}$, it suffices to initialize $k:=r$ instead
  of $k:=1$ in line 3 of Algorithm \ref{algMultiperiodic}.
\end{proof}

\subsubsection{Multiperiodic processes}
\label{secProcesses}

Multiperiodic sequences can be applied to define stationary ergodic
processes. The basic idea is to apply a random initialization, that
is, uniformly random seeds. This leads to the next construction,
called a multiperiodic process. We observe that the existence of the
process is conditional --- given that the value $1$ appears in the seed
sequence infinitely often almost surely.
\begin{definition}[multiperiodic process]
  \label{defiMProcess}
  The multiperiodic process $(K_t)_{t\in\mathbb{Z}}$ with periods
  $(\pi_k)_{k\in\mathbb{N}}$ is the random multiperiodic sequence with
  periods $(\pi_k)_{k\in\mathbb{N}}$ and probabilistically independent
  random seeds $(\Sigma_k)_{k\in\mathbb{N}}$ with uniform
  distributions
  \begin{align}
    \label{Uniform}
    P(\Sigma_k=\sigma_k)=\frac{1}{\pi_k},
    \quad
    \sigma_k\in\klam{1,2,\ldots,\pi_k}
  \end{align}
  provided set $\klam{k\in\mathbb{N}:\Sigma_k=1}$ is infinite almost
  surely.
\end{definition}

Applying the complementary Borel-Cantelli lemmas for independent
events, we derive a simple criterion of existence of a multiperiodic
process.
\begin{theorem}
  The multiperiodic process exists if and only if its periods satisfy
  \begin{align}
    \sum_{k=1}^\infty \frac{1}{\pi_k}=\infty.
  \end{align}
\end{theorem}
\begin{proof}
  Observe that events $(\Sigma_k=1)$ are independent.  Hence by the
  Borel-Cantelli lemmas \cite[Theorems 4.3 and 4.4]{Billingsley79},
  set $\klam{k\in\mathbb{N}:\Sigma_k=1}$ is infinite almost surely if
  and only if
  \begin{align}
    \sum_{k=1}^\infty P(\Sigma_k=1)=\infty.
  \end{align}
  In view of the uniform distributions of
  $(\Sigma_k)_{k\in\mathbb{N}}$, we obtain the claim.
\end{proof}

Moreover, through the approximation by periodic processes, each
multiperiodic process admits a frequency interpretation of its
probability distribution. Namely, it inherits the properties of being
stationary and ergodic from the approximating periodic sources.
\begin{theorem}
  The multiperiodic process $(K_t)_{t\in\mathbb{Z}}$ is stationary and
  ergodic but is not mixing.
\end{theorem}
\begin{proof}
  A sequence $(a_t)_{t\in\mathbb{Z}}$ is called $p$-periodic if
  \begin{align}
    (a_{t+s})_{t\in\mathbb{Z}}=(a_t)_{t\in\mathbb{Z}}\iff s\equiv
    0\mod p.
  \end{align}
  A stochastic process $(A_t)_{t\in\mathbb{Z}}$ is called $p$-periodic
  if
  \begin{align}
    P((A_t)_{t\in\mathbb{Z}}=(a_{t+s})_{t\in\mathbb{Z}})=\frac{1}{p},
    \quad
    s\ge 0,
  \end{align}
  for a certain $p$-periodic sequence $(a_t)_{t\in\mathbb{Z}}$.  It is
  known that $p$-periodic processes are stationary and ergodic but
  they are not mixing \cite[Section 4.1]{Debowski21}.

  Let an $m\in\mathbb{N}$. Let $K^{\le m}_t:=\min\klam{K_t,m}$.
  Observe that process $(K^{\le m}_t)_{t\in\mathbb{Z}}$ depends only
  on seeds $(\Sigma_k)_{k\le m}$. Denote
  $p_m:=\prod_{k=1}^{m}\pi_k$. We will show that process
  $(K^{\le m}_t)_{t\in\mathbb{Z}}$ is $p$-periodic where $p\le p_m$.

  Let $(\Phi_{k,t})_{k\in\mathbb{N},t\in\mathbb{Z}}$ be the hand
  matrix of sequence $(K_t)_{t\in\mathbb{Z}}$.  By recursion
  (\ref{RecursionFirst})--(\ref{RecursionLast}), there is a bijection
  $(\Phi_{k,t})_{k\le m}\leftrightarrow (\Phi_{k,t+1})_{k\le
    m}$. Hence process $(\Phi_{k,t})_{k\le m,t\in\mathbb{Z}}$ is a
  Markov process over time $t$.  There are exactly $p_m$ distinct
  states of hands $(\phi_k)_{k\le m}$ and it can be shown that all
  states communicate. Hence Markov process
  $(\Phi_{k,t})_{k\le m,t\in\mathbb{Z}}$ is irreducible, each state
  $(\phi_k)_{k\le m}$ is visited exactly once within the period $p_m$,
  and the uniform marginal distribution of the initial state
  $(\Phi_{k,1})_{k\le m}=(\Sigma_k)_{k\le m}$ is the invariant
  distribution. In turn, process
  $(\Phi_{k,t})_{k\le m,t\in\mathbb{Z}}$ is stationary, ergodic, and
  $p_m$-periodic. Since each $K^{\le m}_t$ is a function of
  $(\Phi_{k,t})_{k\le m}$, hence process
  $(K^{\le m}_t)_{t\in\mathbb{Z}}$ is $p$-periodic where $p\le p_m$.

  Thus process $(K^{\le m}_t)_{t\in\mathbb{Z}}$ is stationary and
  ergodic but is not mixing.  Since $m\in\mathbb{N}$ is arbitrarily
  chosen, hence also process $(K_t)_{t\in\mathbb{Z}}$ is stationary
  and ergodic but is not mixing. The reason is that verifying
  stationarity, ergodicity, and mixing can be reduced to checking
  probabilities of finite blocks by \cite[Theorem 4.10]{Debowski21} or
  \cite[Lemma 7.15]{Gray09}.
\end{proof}

\subsection{Statistics}
\label{secStatistics}

In this subsection, we analyze several statistics of multiperiodic
sequences and multiperiodic processes. We begin with the relative
frequencies and the marginal distribution. The waiting times are the
second kind of statistics. The third sort of investigated quantities
relate to the number of observed types. Some results concerning these
quantities are established by applying the decimation properties shown
in Theorem \ref{theoDecimated}, the Birkhoff ergodic theorem, the
union bound, and the Borel-Cantelli lemma.

\subsubsection{Relative frequencies}
\label{secFrequencies}

By the decimation property established in Theorem \ref{theoDecimated},
for a multiperiodic sequence there exist limiting relative frequencies.
\begin{theorem}
  \label{theoMarginalSequence}
  Let $(k_t)_{t\in\mathbb{Z}}$ be a multiperiodic sequence with
  periods $(\pi_k)_{k\in\mathbb{N}}$.  The relative frequencies exist
  and equal
  \begin{align}
    \lim_{t\to\infty}\frac{1}{t}\sum_{j=1}^t \boole{k_j\ge k}
    =
    f_k
    :=
    \prod_{j=1}^{k-1}
    \frac{\pi_j-1}{\pi_j}
    .
    \label{Marginal}
  \end{align}
\end{theorem}
\begin{proof}
  Let $(k_t^{\ge r})_{t\in\mathbb{N}}$ be the subsequence of
  $(k_t)_{t\in\mathbb{N}}$ from which we have removed all tokens
  $k_t<r$.  We proceed by induction on $k$. First, we observe
  \begin{align}
    \lim_{t\to\infty}\frac{1}{t}\sum_{j=1}^t \boole{k_j\ge 1}=1=f_1.
  \end{align}
  Second, in view of claim 2.\ of Theorem \ref{theoDecimated}, we
  notice
  \begin{align}
    &\lim_{t\to\infty}\frac{1}{t}\sum_{j=1}^t \boole{k_j\ge k}
      \nonumber\\
    &=
    \okra{\lim_{t\to\infty}\frac{1}{t}\sum_{j=1}^t \boole{k_j\ge k-1}}
    \okra{\lim_{t\to\infty}\frac{1}{t}\sum_{j=1}^t \boole{k_j^{(k-1)}\ge k}}
      \nonumber\\
    &=
    f_{k-1}\okra{\frac{\pi_{k-1}-1}{\pi_{k-1}}}
    =
    f_k.
  \end{align}
  Hence the claim follows.
\end{proof}

Hence by ergodicity of the multiperiodic process, the limiting
relative frequencies are equal to the marginal distribution of the
process.
\begin{theorem}
  \label{theoMarginalProcess}
  The multiperiodic process $(K_t)_{t\in\mathbb{Z}}$ with periods
  $(\pi_k)_{k\in\mathbb{N}}$ has the marginal distribution
  \begin{align}
    \label{MarginalProcess}
    P(K_0=k)=p_k:=f_k-f_{k+1},
  \end{align}
  where $f_k$ is the marginal distribution function (\ref{Marginal}).
\end{theorem}
\begin{proof}
  We obtain distribution (\ref{MarginalProcess}) by the Birkhoff
  ergodic theorem, which establishes the almost sure equality of
  probabilities and relative frequencies for stationary ergodic
  processes \cite{Birkhoff32,Garsia65}.
\end{proof}

\subsubsection{Waiting times}
\label{secTimes}

Another important statistic of a stochastic process is the waiting
time. Following \cite{Kac47, Kontoyiannis98}, we define the waiting
time as the first position in the process where the value of interest
occurs. That is we put
\begin{align}
  W_k:=\inf\klam{t\in\mathbb{N}: K_t=k}.
\end{align}
The celebrated Kac lemma \cite{Kac47} predicts that the conditional
expectation of the waiting time equals the inverse probability of the
recurring value. Namely, we have
\begin{align}
 \mean\okra{W_k|K_0=k}=\frac{1}{P(K_0=k)}
\end{align}
for a stationary ergodic process $(K_t)_{t\in\mathbb{Z}}$.

Complementing this result, for a multiperiodic process, we can obtain
a uniform bound as follows.
\begin{theorem}
  \label{theoWaitingRandom}
  For a multiperiodic process $(K_t)_{t\in\mathbb{Z}}$ with periods
  $(\pi_k)_{k\in\mathbb{N}}$, we have a uniform bound
  \begin{align}
    \label{WaitingRandom}
    W_k
    \le
    w_k:=\sum_{j=2}^{k}\frac{1}{f_j}+\frac{\pi_k}{f_k}
    ,
  \end{align}
  where $f_k$ is the marginal distribution function (\ref{Marginal}).
\end{theorem}
\begin{proof}
  Consider a prefix $(K_1,K_2,\ldots,K_{T_1})$ of the multiperiodic
  process. When we delete all $K_t<r$, the sequence shortens to
  $(K^{\ge r}_1,K^{\ge r}_2,\ldots,K^{\ge r}_{T_r})$. We observe that
  $W_k\le T_1$ if $T_k=\pi_k$, whereas the partial periodicity implies
  \begin{align}
    T_r-T_{r+1}\le \frac{T_r}{\pi_r}+1.
  \end{align}
  Hence we have
  \begin{align}
    T_r\le (T_{r+1}+1)\cdot\frac{\pi_{r}}{\pi_{r}-1}
  \end{align}
  and consequently we obtain
  \begin{align}
    W_k
    &\le
      \okra{\ldots
      \okra{
      \okra{
      \okra{\pi_k+1}
      \frac{\pi_{k-1}}{\pi_{k-1}-1}+1}
      \frac{\pi_{k-2}}{\pi_{k-2}-1}+1}\ldots}
      \frac{\pi_1}{\pi_1-1}
      \nonumber\\
    &=
      \sum_{j=2}^{k}\frac{1}{f_j}+\frac{\pi_k}{f_k}
      =
      w_k
  \end{align}   
  in view of identity
  \begin{align}
    \okra{\ldots\okra{\okra{\okra{a_{k}+1}a_{k-1}+1}a_{k-2}+1}\ldots}a_1
    =\sum_{j=1}^{k}\prod_{i=1}^{j} a_i.
  \end{align}
\end{proof}

Moreover, for increasing periods, we obtain that the waiting times
increase monotonically, whereas the marginal probabilities are
monotonically decreasing.
\begin{definition}[increasing periods]
  Periods $(\pi_k)_{k\in\mathbb{N}}$ are called increasing if
  $\pi_{k+1}\ge\pi_k$ for all $k\in\mathbb{N}$.
\end{definition}
\begin{theorem}
  For increasing periods $(\pi_k)_{k\in\mathbb{N}}$, we have:
  \begin{align}    
    w_{k+1}&\ge w_k,
    \\
    P(K_0=k+1)&\le P(K_0=k).
  \end{align}
\end{theorem}
\begin{proof}
  We have $w_{k+1}\ge w_k$ since
  \begin{align}
    \pi_k+1\le \pi_{k+1}+1\le (\pi_{k+1}+1)\frac{\pi_k}{\pi_k-1}+1.
  \end{align}
  Similarly $P(K_0=k)=f_k-f_{k+1}\ge f_{k+1}-f_{k+2}=P(K_0=k+1)$ since
  \begin{align}
    1-\frac{\pi_k-1}{\pi_k}\ge
    \frac{\pi_k-1}{\pi_k}\okra{1-\frac{\pi_{k+1}-1}{\pi_{k+1}}}.
  \end{align}
\end{proof}

\subsubsection{Number of types}
\label{secTypes}

Now we proceed to bounding the number of types observed in a finite
sample. We prove results of a varying generality. Some of these
propositions hold for general stationary sources and may be useful for
quantitative linguistic applications revolving around studying the
type-token ratio \cite{Baayen01}.

Let us denote the random set of observed types
\begin{align}
  \label{Types}
  \mathcal{V}_t:=\klam{K_1,K_2,\ldots,K_t}
  =\klam{k\in\mathbb{N}:W_k\le t}.
\end{align}
We denote its cardinality, the random number of types in a finite sample,
\begin{align}
  V_t:=\card\mathcal{V}_t.
\end{align}
To bound the number of types effectively, we introduce four other
statistics:
\begin{align}
  u_t&:=\min\klam{k\in\mathbb{N}: w_k>t},
  \\
  U_t&:=\min\klam{k\in\mathbb{N}: k\not\in\mathcal{V}_t},
  \\
  g_t&:=\max\klam{k\in\mathbb{N}:P(K_0=k)\ge\frac{1}{t}},
  \\
  q_t&:=P(K_0>g_t).
\end{align}

The first result is a simple sandwich bound for multiperiodic
processes.
\begin{theorem}
  For a multiperiodic process $(K_t)_{t\in\mathbb{Z}}$, we have
  \begin{align}
    \label{SandwichTypesI}
    u_t-1
    \le
    U_t-1
    \le
    V_t
    .
  \end{align}
\end{theorem}
\begin{proof}
  The inequalities follow by $w_k\ge W_k$ and
  $U_t-1\le \card\mathcal{V}_t$.
\end{proof}

The subsequent two statements pertain to general stationary processes.
We note that functions $g_t$ and $q_t$ were introduced by Khmaladze
\cite{Khmaladze88} for memoryless sources.  Generalizing his technique
from memoryless processes to stationary ones, we can prove this
proposition.
\begin{theorem}
  For a stationary process $(K_t)_{t\in\mathbb{Z}}$ over natural
  numbers, we have
  \begin{align}
    \label{SandwichTypesII}
    \mean V_t
    \le
    g_t+tq_t.
  \end{align}
\end{theorem}
\begin{proof}
  We observe
  \begin{align}
    \mean V_t
    =
    \mean \sum_{k=1}^\infty \boole{k\in\mathcal{V}_t}
    =
    \sum_{k=1}^\infty P(k\in\mathcal{V}_t)
    \le
    g_t+\sum_{k>g_t} P(k\in\mathcal{V}_t)
  \end{align}
  and the union bound for stationary processes
  \begin{align}
    P(k\in\mathcal{V}_t)
    =
    P(K_1=k\lor \ldots\lor K_t=k)
    \le
    tP(K_0=k).
  \end{align}
  Hence the claim follows.
\end{proof}

Quotient $V_t/t$ is called the type-token ratio in quantitative
linguistics \cite{Herdan64, Baayen01}. It is easy to see
that it tends to zero in the stationary case.
\begin{theorem}
  For a stationary process $(K_t)_{t\in\mathbb{Z}}$ over natural
  numbers, 
  \begin{align}
    \label{TypeTokenExp}
    \lim_{t\to\infty}\frac{\mean V_t}{t}&=0.
  \end{align}
\end{theorem}
\begin{proof}
  Applying the ideas from the proof of the previous theorem, we
  observe that bound
  \begin{align}
    \label{SandwichTypesIII}
    \mean V_t \le h_t+tP(K_0>h_t)
  \end{align}
  holds for any function $h_t$. In particular, for an $\epsilon>0$, we
  may take $h_t=\epsilon t/2$. For all sufficiently large $t$, we
  observe $P(K_0>h_t)\le \epsilon/2$. Hence, for these $t$, we have
  $\mean V_t/t \le h_t/t+P(K_0>h_t)\le \epsilon$. By arbitrariness of
  $\epsilon$, we derive (\ref{TypeTokenExp}).
\end{proof}

We will need a stronger bound that holds under a mild condition. 
\begin{theorem}
  For a stationary process $(K_t)_{t\in\mathbb{Z}}$ over natural
  numbers such that $\mean (\log K_0)^3<\infty$, we have
  \begin{align}
    \label{TypeTokenLogExp}
    \lim_{t\to\infty}\frac{\mean V_t\log t}{t}&=0,
    \\
    \label{TypeTokenLogAS}
    \lim_{t\to\infty}\frac{V_t\log t}{t}&=0 \text{ a.s.}
  \end{align}
\end{theorem}
\begin{proof}
  Observe by the Markov inequality that 
  \begin{align}
    \sum_{k=1}^\infty kP\okra{K_0\ge \frac{2^k}{k^3}}
    &\le
      \sum_{k=1}^\infty kP\okra{(\log K_0)^3\ge (k-3\log k)^3}
      \nonumber\\
    &\le
    \mean (\log K_0)^3 \sum_{k=1}^\infty \frac{k}{(k-3\log k)^3}<\infty.
  \end{align} 
  Choose $h_t=t/(\log t)^3$ in (\ref{SandwichTypesIII}). Then we obtain
  \begin{align}
    \frac{\mean V_{2^k}k}{2^k}\le \frac{1}{k^2}+kP\okra{K_0\ge \frac{2^k}{k^3}}.
  \end{align}
  Consequently, by the Markov inequality, we have
  \begin{align}
    \sum_{k=1}^\infty P\okra{\frac{V_{2^k}k}{2^k}\ge \epsilon}
    &\le
      \frac{1}{\epsilon}\sum_{k=1}^\infty \frac{\mean V_{2^k}k}{2^k}
      \nonumber\\
    &\le
      \frac{1}{\epsilon}\sum_{k=1}^\infty
      \okra{\frac{1}{k^2} +kP\okra{K_0\ge \frac{2^k}{k^3}}}<\infty.
  \end{align}
  Hence $\lim_{k\to\infty} \mean V_{2^k}k/2^k=0$, whereas by the
  Borel-Cantelli lemma \cite[Theorems 4.3]{Billingsley79}, we derive
  $V_{2^k}k/2^k<\epsilon$ for all but finitely many $k$ almost
  surely. Notice now that $V_t$ is a growing function of $t$. Hence
  $(V_t\log t)/t\le 2V_{2^k}k/2^k$ for $2^{k-1}\le t\le 2^k$. In
  particular, we have (\ref{TypeTokenLogExp}) and
  $\limsup_{t\to\infty} (V_t\log t)/t\le 2\epsilon$ almost surely.  By
  arbitrariness of $\epsilon$, we derive (\ref{TypeTokenLogAS}).
\end{proof}

\subsection{Information}
\label{secInformation}

In this subsection, we demonstrate propositions that are central from our
particular perspective. The first result concerns the flow of
information in the description of a sample drawn from a multiperiodic
process.  We show that a sample of a multiperiodic process carries
exactly the same information as the set of seeds for types appearing
in the sample. These two objects can be computed from one another and
hence their entropies are equal. Consequently, this result allows to
sandwich bound the block entropy of a multiperiodic process by the
number of types and related statistics. In particular, Hilberg's law
follows from Heaps' law, cf.\ \cite{Debowski25g} for a similar result
for Santa Fe processes.

\subsubsection{Seed estimator}
\label{secSeeds}

Let us analyze the relationship between a finite multiperiodic sample
and the pool of random seeds in more detail. For bounding the entropy
of the sample, it pays off to state this proposition explicitly.
\begin{theorem}
  \label{theoSeeds}
  Consider a multiperiodic process $(K_t)_{t\in\mathbb{Z}}$. 
  We claim that:
  \begin{enumerate}
  \item Set $\klam{(k,\Sigma_k): k\in\mathcal{V}_t}$ is a
    function of sample $K_1^t$.
  \item Define the seed estimator
    \begin{align}
      \hat\Sigma_k:=
      \begin{cases}
        \Sigma_k, & k\in\mathcal{V}_t,
        \\
        \pi_k, & k\not\in\mathcal{V}_t.
      \end{cases}
    \end{align}
    Let $(\hat K_t)_{t\in\mathbb{Z}}$ be the multiperiodic sequence with
    seeds $(\hat\Sigma_k)_{k\in\mathbb{N}}$. Then
    \begin{align}
      K_1^t=\hat K_1^t.
    \end{align}
  \item Sample $K_1^t$ is a function of set
    $\klam{(k,\Sigma_k): k\in\mathcal{V}_t}$.
  \item If $k\in\mathcal{V}_t$ then $k\le M_t$ and
    $\Sigma_k\le t$, where
    \begin{align}
    M_t:=\max\mathcal{V}_t.
    \end{align}
  \end{enumerate}
\end{theorem}
\begin{proof}
  The subsequent claims are dealt with below.
  \begin{itemize}
  \item Let $L_k=\min\klam{t\ge 1:K_t^{(k)}=k}$, where
    $(K^{(k)}_t)_{t\in\mathbb{N}}$ is the subsequence of sequence
    $(K_t)_{t\in\mathbb{N}}$ from which we have removed all tokens
    $k_t<k$.  If $k\in\mathcal{V}_t$ then the first $L_k$ tokens of
    $(K^{(k)}_t)_{t\in\mathbb{N}}$ are a function of sample $K_1^t$.
    By claim 2.\ of Theorem \ref{theoDecimated}, $L_k=\Sigma_k$. Hence
    we obtain claim 1.
  \item We observe that altering $\Sigma_k$ for
    $k\not\in\mathcal{V}_t$ to any value in set
    \begin{align*}
      \klam{\Sigma_k,\Sigma_k+1,\ldots,\pi_k}
    \end{align*}
    does not change the resulting multiperiodic sample $K_1^t$. It is
    so since any such change only delays the occurrence of token $k$
    beyond the scope of $K_1^t$, leaving the tokens within $K_1^t$
    intact. Hence claim 2.\ follows.
  \item Obviously, claim 2.\ implies claim 3.
  \item If $k\in\mathcal{V}_t$ then $k\le\max\mathcal{V}_t$ and
    $\Sigma_k=L_k\le t$ since the first $L_k$ tokens of sequence
    $(K^{(k)}_t)_{t\in\mathbb{N}}$ are a subsequence of sample
    $K_1^t$. Hence we obtain claim 4. \qedhere
  \end{itemize}
\end{proof}

\subsubsection{Block entropy}
\label{secEntropy}

Now, we will sandwich bound the Shannon entropy in terms of the number
of types.  For this goal, we will use a bound stemming from the
H\"older and Jensen inequalities. Let $\ln$ denote the natural
logarithm.
\begin{lemma}
  \label{theoHolderLog}
  Suppose that $X\ge 0$ and $Y\ge 1$. Then for any $p>1$, we have
  \begin{align}
    \mean X\ln Y\le 
    (\mean X^p)^{1/p}\okra{\ln\mean Y+\frac{1}{p-1}}.
  \end{align}
\end{lemma}
\begin{proof}
  Let $p>1$ and $q=p/(p-1)$. For $X,Z\ge 0$, the H\"older
  inequality is
  \begin{align}
    \mean XZ\le (\mean X^p)^{1/p}(\mean Z^q)^{1/q}.
  \end{align}
  The second derivative of function $y\mapsto(\ln y)^q$ is negative
  for $y>e^{q-1}$. Hence by the Jensen inequality (J), we obtain
  \begin{align}
    \mean(\ln Y)^q\le\mean(\ln(e^{q-1}Y))^q
    \stackrel{(J)}{\le}
    (\ln(e^{q-1}\mean Y))^q\le\okra{\ln\mean Y+\frac{1}{p-1}}^q.
  \end{align}
  Using $Z=\ln Y$ in the H\"older inequality and chaining it with the
  above chain of inequalities, we derive the claim.
\end{proof}

Subsequently, we bound the entropy. Here $\log$ is the base 2
logarithm.
\begin{theorem}
  \label{theoEntropyBound}
  For a multiperiodic process $(K_t)_{t\in\mathbb{Z}}$ and for any
  $q\ge 1$,
  \begin{align}
  \label{EntropyBound}
    u_t-1
    \le
    H(K_1^t)
    \le
    \okra{e\mean V_t+3}
    \okra{(q+2)\log t+q\log\mean K_0^{1/q}}+2.
  \end{align}
\end{theorem}
\begin{proof}
  By claim 1.\ of Theorem \ref{theoSeeds} and property
  $k<u_t\implies k\in\mathcal{V}_t$, the tuple of seeds
  $\Sigma_1^{u_t-1}$ is a function of sample $K_1^t$. Hence we obtain
  the lower bound
  \begin{align}
    H(K_1^t)\ge H(\klam{(k,\Sigma_k): k\in\mathcal{V}_t})\ge
    H(\Sigma_1^{u_t-1})\ge u_t-1,
  \end{align}
  since seeds are independent and uniformly distributed over two or
  more values.  As for the upper bound, by claims 3.\ and 4.\ of
  Theorem \ref{theoSeeds}, we obtain
\begin{align}
  H(K_1^t)
  &\le
  H(\klam{(k,\Sigma_k): k\in\mathcal{V}_t})
  \nonumber\\
  &\le H(M_t)+H(\klam{(k,\Sigma_k): k\in\mathcal{V}_t}|M_t)
  \nonumber\\
  &\le 2\mean\log(M_t+1)+\mean (V_t+1)(\log M_t+\log t)
  \nonumber\\
  &\le 2\mean\log M_t+\mean (V_t+1)(\log M_t+\log t)+2.
\end{align}

Next, we observe $1\le M_t^{1/q}\le \sum_{j=1}^t K_j^{1/q}$. Hence
$\mean M_t^{1/q}\le t\mean K_0^{1/q}$ by stationarity. Thus, by
Lemma~\ref{theoHolderLog}, for any $p>1$, we may write
\begin{align}
  H(K_1^t)
  &\le 2\mean q\log M_t^{1/q}+
    \mean (V_t+1)\okra{q\log M_t^{1/q}+\log t}+2
    \nonumber\\
  &\le [(\mean V_t^p)^{1/p}+3]
    \okra{q\log\mean M_t^{1/q}+\log t+\frac{\log e}{p-1}}+2
    \nonumber\\
  &\le [(\mean V_t^p)^{1/p}+3]
    \okra{(q+1)\log t+q\log\mean K_0^{1/q}+\frac{\log e}{p-1}}+2
    .  
\end{align}

Let us substitute $p=1+1/\ln t$ and apply inequality
$V_t^p\le t^{p-1} V_t$. Then we may rewrite
\begin{align}
  H(K_1^t)
  &\le \okra{(t^{p-1}\mean V_t)^{1/p}+3}
    \okra{(q+1)\log t+q\log\mean K_0^{1/q}+\frac{\log e}{p-1}}+2
    \nonumber\\
  &= \okra{(e\mean V_t)^{\frac{\ln t}{\ln t+1}}+3}
    \okra{(q+2)\log t+q\log\mean K_0^{1/q}}+2
    \nonumber\\
  &\le \okra{e\mean V_t+3}
    \okra{(q+2)\log t+q\log\mean K_0^{1/q}}+2,
\end{align}
which is the claimed inequality.
\end{proof}

Hence we have a criterion for the vanishing Shannon entropy rate.
\begin{theorem}
  Suppose that $\mean K_0^{1/q}<\infty$ for a $q\ge 1$. Then the
  multiperiodic process $(K_t)_{t\in\mathbb{Z}}$ satisfies
  \begin{align}
    \label{ZeroEntropy}
    \lim_{t\to\infty}\frac{H(K_1^t)}{t}&=0.
  \end{align}
\end{theorem}
\begin{proof}
  The claim follows by (\ref{TypeTokenLogExp}) and
  (\ref{EntropyBound}) since $\mean K_0^{1/q}<\infty$ for a $q\ge 1$
  implies $\mean (\log K_0)^3<\infty$.
\end{proof}

\subsection{Examples}
\label{secExamples}

Theorem \ref{theoEntropyBound} implies that Hilberg's law
(\ref{Hilberg}) for a multiperiodic process essentially follows from
Heaps' law (\ref{Heaps}).  In this subsection, we will work out two
concrete examples of processes.  First, we will analyze the
multiperiodic process with all periods equal to a fixed constant. This
case leads to the number of types that grows logarithmically with the
sample size.  Second, we will inspect the multiperiodic process with
periods that grow linearly. This case leads to the number of types
that grows like a power law, Heaps' law (\ref{Heaps})
indeed. Consequently, for this example, we also obtain Hilberg's law
(\ref{Hilberg}).

\subsubsection{Constant periods}
\label{secConstant}

First, let us consider the multiperiodic process where periods equal
to a fixed constant. To be concrete, let $\pi_k=p$ for all
$k\in\mathbb{N}$. Then we obtain a geometric distribution of
$P(K_0=k)=p_k$. In particular, we may bound
\begin{align}
  f_k
  &=\prod_{j=1}^{k-1}\frac{\pi_j-1}{\pi_j}
    =
    \okra{\frac{p-1}{p}}^{k-1},
  \\
  p_k
  &=f_k-f_{k+1}=\frac{1}{\pi_k}\prod_{j=1}^{k-1}\frac{\pi_j-1}{\pi_j}
     =
    \frac{1}{p}\okra{\frac{p-1}{p}}^{k-1},
  \\
  w_k
  &=\sum_{j=2}^{k}\frac{1}{f_j}+\frac{\pi_k}{f_k}
    =
    \sum_{j=1}^{k-1} \okra{\frac{p}{p-1}}^j+p\okra{\frac{p-1}{p}}^{k-1}
    \nonumber\\
  &=
    \frac{\okra{\frac{p}{p-1}}^k-1}{\frac{p}{p-1}-1}-1+
    \okra{\frac{p}{p-1}}^{k-1}p
    =
    2p\okra{\frac{p}{p-1}}^{k-1}-1.
\end{align}
Hence we may evaluate
\begin{align}
  u_t
  &=\min\klam{k:w_k>t}
    =
    \min\klam{k:2p\okra{\frac{p}{p-1}}^{k-1}>t+1}
    \nonumber\\
  &=
    \ceil{\frac{\log(t+1)-\log p-1}{\log\okra{\frac{p}{p-1}}}}+2,
  \\
  g_t
  &=\max\klam{k:p_k^{-1}\le t}
    =
    \max\klam{k:p\okra{\frac{p}{p-1}}^{k-1}\le t}
    \nonumber\\
  &=
    \floor{\frac{\log(t+1)-\log p}{\log\okra{\frac{p}{p-1}}}}+1,
  \\
  \log q_t
  &=\log f_{g_t+1}
    = - g_t\log \okra{\frac{p}{p-1}}
    \le -\log(t+1)+\log p.
\end{align}
Hence bound $u_t-1\le \mean V_t\le g_t+tq_t$ can be rewritten as
\begin{align}
  \ceil{\frac{\log(t+1)-\log p-1}{\log\okra{\frac{p}{p-1}}}}+1
  \le \mean V_t\le 
  \floor{\frac{\log(t+1)-\log p}{\log\okra{\frac{p}{p-1}}}}+1+p.
\end{align}
In this case, the Shannon entropy rate vanishes but we do not have
Hilberg's law (\ref{HilbergZero}) --- by Theorem \ref{theoEntropyBound}.

\subsubsection{Linear periods}
\label{secLinear}

Now, we consider the multiperiodic process where periods grow
approximately in a linear fashion. To be concrete, let
$\pi_k=1+\ceil{ck}$ for a $c>0$ and all $k\in\mathbb{N}$.  We may
upper bound
\begin{align}
  \prod_{j=1}^{k-1}\frac{\pi_j-1}{\pi_j}
  &=
    \prod_{j=1}^{k-1} \okra{1-\frac{1}{1+\ceil{cj}}}
    \le
    \exp\okra{\sum_{j=1}^{k-1}\ln\okra{1-\frac{1}{cj+2}}}
    \nonumber\\
  &\le
    \exp\okra{-\sum_{j=1}^{k-1}\frac{1}{cj+2}}
    \le
    \exp\okra{-\int_1^{k}\frac{dx}{cx+2}}
    \nonumber\\
  &=
    \exp\okra{-\frac{1}{c}\ln\frac{ck+2}{c+2}}=\okra{\frac{c+2}{ck+2}}^{1/c}.
\end{align}
Similarly, we may bound conversely
\begin{align}
  \prod_{j=1}^{k-1}\frac{\pi_j}{\pi_j-1}
  &=
    \prod_{j=1}^{k-1} \okra{1+\frac{1}{\ceil{cj}}}
    \le
    \exp\okra{\sum_{j=1}^{k-1}\ln\okra{1+\frac{1}{cj}}}
    \nonumber\\
  &\le
    \exp\okra{\sum_{j=1}^{k-1}\frac{1}{cj}}
    \le
    \exp\okra{\frac{1}{c}+\int_1^{k}\frac{dx}{cx}}
    \nonumber\\
  &=
    \exp\okra{\frac{1}{c}+\frac{1}{c}\ln k}=(ek)^{1/c}.
\end{align}

In consequence, we obtain Zipf's law (\ref{Zipf}) for
$P(K_0=k)=p_k$. In particular, we may bound
\begin{align}
  f_k
  &=\prod_{j=1}^{k-1}\frac{\pi_j-1}{\pi_j}
  \ge (ek)^{-1/c},
  \\
  f_k
  &=\prod_{j=1}^{k-1}\frac{\pi_j-1}{\pi_j}
  \le \okra{\frac{c+2}{ck+2}}^{1/c}
  \le \okra{\frac{c+2}{c}}^{1/c} k^{-1/c},
  \\
  p_k
  &=\frac{1}{\pi_k}\prod_{j=1}^{k-1}\frac{\pi_j-1}{\pi_j}
    \ge \frac{1}{ck+1}(ek)^{-1/c}\ge \frac{e^{-1/c}}{c+2} k^{-(c+1)/c},
  \\
  p_k
  &=\frac{1}{\pi_k}\prod_{j=1}^{k-1}\frac{\pi_j-1}{\pi_j}
  \le \frac{1}{ck+1}\okra{\frac{c+2}{ck+2}}^{1/c}
    \le \frac{1}{c}\okra{\frac{c+2}{c}}^{1/c} k^{-(c+1)/c}.
\end{align}
Thus we may bound the waiting times
\begin{align}
  w_k
  &=\sum_{j=2}^{k}\frac{1}{f_j}+\frac{\pi_k}{f_k}
  \le
  \sum_{j=2}^{k} (ej)^{1/c}+(ck+2)(ek)^{1/c}
  \nonumber\\
  &\le 
  \int_0^k (ex)^{1/c} dx+(ek)^{1/c}+(ck+2)(ek)^{1/c}
  \nonumber\\
  &=\okra{\frac{ck}{c+1}+ck+3}(ek)^{1/c}
  \le \okra{c+\frac{c}{c+1}+3}e^{1/c}k^{(c+1)/c}.
\end{align}

Hence we may bound further
\begin{align}
  u_t
  &=\min\klam{k:w_k>t}
    \ge
    \okra{\frac{t}{\okra{c+\frac{c}{c+1}+3}e^{1/c}}}^{c/(c+1)},
  \\
  g_t
  &=\max\klam{k:p_k^{-1}\le t}
   \le
   \okra{\frac{t}{c}\okra{\frac{c+2}{c}}^{1/c}}^{c/(c+1)},
  \\
  g_t
  &=\max\klam{k:p_k^{-1}\le t}
    \ge
    \okra{\frac{te^{-1/c}}{c+2}}^{c/(c+1)}
\end{align}
so as to derive
\begin{align}
  q_t&= f_{g_t+1}
       \le \okra{\frac{c+2}{c}}^{1/c}
       \okra{\frac{te^{-1/c}}{c+2}}^{-1/(c+1)}.
\end{align}

As a result, bound $u_t-1\le \mean V_t\le g_t+tq_t$ implies Heaps' law
\begin{align}
  C_1t^{c/(c+1)}\le \mean V_t\le C_2 t^{c/(c+1)}
\end{align}
for certain constants $C_1,C_2>0$. By Theorem \ref{theoEntropyBound},
to derive Hilberg's law (\ref{HilbergZero}) with $\beta=c/(c+1)$, it
suffices to show that $\mean K_t^{1/q}<\infty$ for some $q\ge 1$. In
fact, for any $q>c$, we obtain
\begin{align}
  \mean K_t^{1/q}
  &=\int_0^\infty P(K_t^{1/q}\ge k)dk
    =\int_0^\infty P(K_t\ge k^q) dk
    \nonumber\\
  &=\int_0^\infty f_{\ceil{k^q}} dk
    \le
    \int_1^\infty \okra{\frac{c+2}{c}}^{1/c} k^{-q/c}dk<\infty.
\end{align}

\begin{figure}[t]
  \centering
  \includegraphics[width=0.9\textwidth]{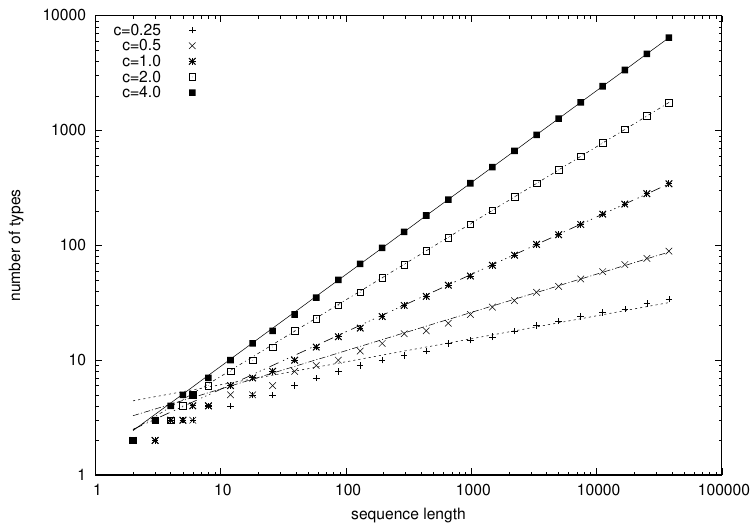}
  \caption{\label{figTypes} The number of types $V_t$ as a function of
    sequence length $t$ for periods $\pi_k=1+\ceil{ck}$ and seeds
    $\sigma_k=1$. The lines are theoretical predictions
    $V_t\sim t^{c/(c+1)}$, where the proportionality constant is
    chosen by least squares.}
\end{figure}

To provide some visualization of these results, in Figure
\ref{figTypes}, we present the growth of the actual number of types
$V_t$ for three multiperiodic sequences with periods
$\pi_k=1+\ceil{ck}$ and seeds $\sigma_k=1$, where
$c\in\klam{0.25,0.5,1.0,2.0,4.0}$. We notice a good agreement with the
theoretical prediction $\mean V_t\sim t^{c/(c+1)}$.

\section{Conclusion}
\label{secConclusion}

We have introduced a new class of stationary stochastic processes,
called multiperiodic processes, generated by random shifts of
deterministic multiperiodic sequences. Under suitable choices of the
period parameters, these processes are stationary and ergodic, though
not mixing, and they satisfy Hilberg's law with vanishing entropy rate
\cite{Hilberg90}.  The central mechanism of the construction is an
infinite collection of independent random seeds that determine the
shifts of the underlying multiperiodic structure. In contrast to the
latent variables of Santa Fe processes
\cite{Debowski09,Debowski12,Hutter21}, these seeds are measurable with
respect to the sample path but not with respect to the shift-invariant
algebra. This distinction makes it possible to combine ergodicity with
power-law block entropy growth.

From a broader perspective, multiperiodic processes complement the
existing collection of examples related to Hilberg's law. Santa Fe and
Oracle processes provide simple constructions with positive entropy
rate, whereas random hierarchical association processes (RHA)
processes \cite{Debowski17,Debowski21} exhibit vanishing entropy rate
but are considerably more complicated. The multiperiodic construction
occupies a different point in this landscape by combining an
elementary description, ergodicity, and zero entropy rate.  Although
multiperiodic processes are unlikely to be useful models of natural
language because of their lack of mixing, they provide analytically
tractable benchmark sources whose entropy rate and Hilberg exponent
can be controlled. Such examples may be useful for testing entropy
estimators, universal coding schemes, and learning algorithms on data
with long-range statistical dependencies.

We have also identified an open question. Namely, it would be
desirable to construct a simple class of processes that simultaneously
exhibits Hilberg's law, vanishing entropy rate, and strong mixing. The
absence of such examples in the current road map, summarized in
Section \ref{secMap}, suggests that this problem may require new ideas.

\section*{Prior versions}

This work stems from rethinking and heavily editing an imperfect
manuscript \cite{Debowski23}, which was posted to ArXiv on February
17, 2023.  That manuscript has not been published anywhere else but
was presented as a seminar talk at the Math Machine Learning seminar
in the Max Planck Institute for Mathematics in the Sciences.

\section*{AI application}

The abstract, the introduction, Section \ref{secContext}, and the
conclusion have been sketched in a dialog between the author and
ChatGPT (\url{chatgpt.com}).  The suggestions of the language model
have been post-edited.

\section*{Funding}

This work received no external funding.

\section*{Acknowledgment}

I thank Jan Mielniczuk, Szymon Jaroszewicz, Marcus Hutter, and Guido
Mont\'ufar for discussing prior versions of this paper.

\bibliographystyle{abbrvnat}

\setlength{\bibsep}{2pt}
\bibliography{0-journals-abbrv,0-publishers-abbrv,ai,ql,nlp,math,mine,tcs,books}

\end{document}